\documentclass[english,aps,prl,onecolumn,notitlepage]{revtex4-1}
\usepackage[T1]{fontenc}
\usepackage[latin9]{inputenc}

\expandafter\let\csname equation*\endcsname\relax
\expandafter\let\csname endequation*\endcsname\relax
\usepackage{amsmath}
\usepackage{babel}
\usepackage{mathrsfs}
\usepackage{bm}
\usepackage{amsthm}
\usepackage{amsopn}
\usepackage{stmaryrd}
\usepackage[unicode=true,
 bookmarks=false,
 breaklinks=false,pdfborder={0 0 1},backref=false,colorlinks=false,hidelinks]
 {hyperref}

\makeatletter

\usepackage{color}
\RequirePackage[dvipsnames,usenames]{xcolor}

\DeclareMathOperator*{\argmin}{arg\,min}

\DeclareMathOperator*{\supp}{supp}

\renewcommand{\eqref}[1]{Eq.~\ref{#1}}

\newif\ifnotes
\notestrue

\ifnotes
  \newcommand{\dhwc}[1]{{\color{Red}{{DHW COMMENT: #1}}}}
 \newcommand{\dhws}[1]{{\color{blue}{{DHW SUGGESTION: #1}}}}
 \newcommand{\ak}[1]{{\color{ForestGreen}{{AK: #1}}}}
 
\else
 \newcommand{\dhwc}[1]{}
 \newcommand{\dhws}[1]{}
 \newcommand{\ak}[1]{}
\fi

\newtheorem{thm}{\protect\theoremname}

\providecommand{\theoremname}{Theorem}

\newcommand{\proc}{\mathcal{P}}
\newcommand{\imgPhi}{\mathcal{T}}

\newtheorem{theorem}{Theorem}

\newtheorem{proposition}[theorem]{Proposition}

\makeatother

\begin{document}

\title{Dependence of dissipation on the initial distribution over states}

 \author{Artemy Kolchinsky}
  \address{Santa Fe Institute, 1399 Hyde Park Road, Santa Fe, NM 87501, USA}

 \author{David H. Wolpert}
  \address{1 Santa Fe Institute, 1399 Hyde Park Road, Santa Fe, NM 87501, USA}
\address{2 Massachusetts Institute of Technology, 77 Massachusetts Ave, Cambridge, MA 02139, USA}
 \address{3 Arizona State University, Tempe, AZ 85281, USA}

\begin{abstract}
We analyze how the amount of work dissipated by a fixed nonequilibrium process depends on the initial distribution over states.
Specifically, we compare the amount of dissipation when the process is used with some specified initial distribution to the minimal amount 
of dissipation possible for any initial distribution. 
We show that the difference between those two amounts of dissipation is given by a simple information-theoretic function that depends only on the initial and final state distributions. Crucially, this difference is independent of the details of the process
relating those distributions. We then consider how dissipation depends on the initial distribution for a `computer', i.e., 
a nonequilibrium process whose dynamics over coarse-grained macrostates implement some desired input-output map. We 
show that our results still apply when stated in terms of distributions over the computer's coarse-grained macrostates. 
This can be viewed as a novel thermodynamic cost of computation, reflecting changes in the distribution over inputs
rather than the logical dynamics of the computation.
\end{abstract}

\maketitle

\section{Introduction}
The past few decades have seen great advances in nonequilibrium statistical physics~\cite{touchette2004information,sagawa2009minimal,dillenschneider2010comment,sagawa2012fluctuation,crooks1999entropy,crooks1998nonequilibrium,chejne2013simple,jarzynski1997nonequilibrium,esposito2011second,esposito2010three,parrondo2015thermodynamics,sagawa2009minimal,pollard2016second,wiesner2012information,still2012thermodynamics,prokopenko2013thermodynamic,prokopenko2014transfer},
resulting in many novel predictions and experiments~\cite{dunkel2014thermodynamics,roldan2014universal,berut2012experimental}. 
Some of the most important results of this research
have been powerful new tools for analyzing the \emph{dissipated work} (or ``dissipation'' for short) in nonequilibrium 
processes.
Dissipation is the amount of work done on an evolving system that exceeds the theoretical minimal amount
needed to drive such a system from its initial to its final distribution~\cite{esposito2011second,hasegawa2010generalization,takara2010generalization,deffner_information_2012,parrondo2015thermodynamics}.  
Equivalently, it is proportional to the (irreversible) entropy production during the course of the process, i.e., the total change in entropy
of the system minus the amount of entropy that flows from the heat bath to the system in the form of heat.

Several expressions for the amount of dissipation in any given process have been derived by exploiting the detailed fluctuation theorems (DFTs)~\cite{crooks1998nonequilibrium,crooks1999entropy,seifert_entropy_2005,jarzynski_rare_2006,seifert2012stochastic},
typically under the assumption of dynamics that obeys local detailed balance.
These results express the dissipation in terms of the Kullback-Leibler (KL) divergence \cite{cover_elements_2012,mack03}
between the probability density over state trajectories occurring in the original process
and the probability density %
under a special `time-reversed' 
version of the process.
However these results 
are impractical for quantifying dissipation in many cases of interest, since computing the KL divergence requires integration 
of a probability density over all possible trajectories.

Related research has investigated lower bounds on dissipation by studying \emph{optimal processes}.
These are processes that achieve minimal dissipation subject to some specified set of constraints~\cite{takara2010generalization,hasegawa2010generalization,parrondo2015thermodynamics}.
For example, optimal processes have been identified for transforming some desired initial Hamiltonian and state distribution into a different Hamiltonian and state distribution
under a finite-time constraint~\cite{schmiedl2007optimal,sivak2012thermodynamic,aurell2012refined}, or while obeying a constraint on allowable work fluctuations~\cite{funo_work_2016}. 
Some authors have also considered how changes to the initial distribution affect the work and dissipation 
\emph{if the process is  changed to be optimal for the new distribution}~\cite{hasegawa2010generalization,sagawa2009minimal,dillenschneider2010comment}.
Such research is concerned with processes that minimize dissipation, and more generally with how dissipation varies with changes to the process.

Here we consider a complementary problem, which to our knowledge has never been previously
analyzed. We suppose that there is a fixed process $\proc$, coupled to a heat bath that is at a constant temperature.
We then consider a very common real-world scenario, in which this same process
can be run with different initial distributions over states.
We ask, how does the the amount of work dissipated by $\proc$ vary with changes to the initial distribution? 
What is the maximal cost in extra dissipation that can arise by using one initial distribution rather than another?
 How do these answers depend on the details of the process $\proc$?

Surprisingly, we find that the dependence of dissipation on the initial distribuiton has a simple information-theoretic form. Let $q_0$ be an initial distribution over the states for which $\proc$ dissipates the minimal amount of work. We 
prove that the dissipation arising from using some arbitrary initial
distribution $r_0$ is the dissipation arising from using $q_0$, plus the reduction of
the Kullback-Leibler (KL) divergence between $r_0$ and $q_0$ from the beginning to the end of $\proc$.
The additional dissipation incurred when $\proc$ is initialized with $r_0 \ne q_0$ is independent
of all intermediate details of \emph{how} $\proc$ changes the initial distribution into the final one.

Our analysis provides a useful and novel tool for calculating dissipated work for a given thermodynamic
process run on a given initial distribution. For example,
suppose we design a process to be dissipationless (i.e., thermodynamically reversible) when run with some initial state distribution. Our analysis can be used to calculate exactly how much work would be dissipated if that process were run with some other state distribution.  %
As a demonstration (\ref{appendix:mandal-and-jarzynski}), we consider a published model of Maxwell's demon~\cite{mandal2012work}, a device that extracts work from an incoming stream of bits, and compute dissipation as a function of the distribution of bits. 

More generally, consider a fixed process connected to a heat bath that is at a constant temperature, which dissipates least work when prepared with some particular initial distribution.  For example, this might be a process in which a volume of gas expands while pushing against a piston and lifting a weight.  There will be some `optimal' initial distribution of the states of the gas which minimizes dissipated work.  Our results state how much more work will be dissipated when the gas is prepared with some other initial distribution.

After deriving this result, we extend it to analyze dissipation in a physical computer. More precisely, 
we suppose that there is a coarse-graining of the states of our system into a set of macrostates.  These macrostates are identified with logical values and the dynamics over the macrostates
identified with the (possibly noisy) computation. 
The initial distribution over the macrostates may reflect how a user of the computer initializes its logical values.
As before, we consider how the additional dissipation incurred by a computer, above and beyond the minimum, depends on the initial macrostate distribution.
We show that the additional dissipation is still given by the drop in KL divergence, only now stated in terms of distributions over the macrostates.

To illustrate the implications of this result for thermodynamics of computation, 
suppose we construct a process that performs a given computation, and that achieves zero dissipation for \emph{some} initial distribution over its 
states (e.g., when employed by one particular user of that computer). Our results quantify how much computer will dissipate if it is instead initialized according to a different distribution (e.g., if the computer is employed by some other user).

We emphasize that these results are equalities, not just bounds.  Furthermore, our results 
give dissipation in terms of a difference in two KL divergences, concerning initial and final \emph{state distributions}.
Thus they  differ fundamentally from previously derived DFTs, which give dissipation in terms of
a single KL divergence, concerning forward and time-reversed \emph{trajectory distributions}.  Moreover, in contrast to such DFTs, our results do not assume local detailed balance.

Our analysis of dissipated work should also be distinguished from earlier analyses of \emph{reversible work}, in particular
analyses expressing reversible work as a reduction of KL divergence between nonequilibrium and equilibrium distributions at initial and final times plus the difference of equilibrium free energies~\cite{takara2010generalization}.  Reversible work is the work required to perform a given transformation using an optimal process, and can be thermodynamically recovered by reversing the process. Dissipated work, on the other hand, is work that is irreversibly lost as entropy production. Furthermore, in general the KL divergences that arise in our analysis do not necessarily involve equilibrium
distributions.

There is one previously-known result that is a special case of our analysis: if a system is prepared with some nonequilibrium distribution $r_0$ and then 
undergoes a non-driven process in which it fully relaxes to equilibrium, the dissipated work is equal to the KL divergence between $r_0$ and the equilibrium distribution~\cite{kawai_dissipation:_2007}.  Our analysis generalizes this earlier result significantly, allowing for processes that do not relax fully to equilibrium. It also applies to processes that are driven by an external work reservoir, in which the equilibrium changes over time, during which the system can remain arbitrarily far from equilibrium at all times.

\section{Formal Background}

We consider
a physical system with a countable set of microstates $X$ that
evolves across a countable set of times $t \in \{0, \Delta \tau, 2 \Delta \tau, \ldots, 1\}$, 
while in contact with a heat bath at temperature $T$. 
We use $x_{0..1} := ( x_0, x_{\Delta \tau}, \ldots, x_1 )$ to indicate a particular trajectory through the system's state space.
The system may also be connected to a work reservoir throughout its evolution, which causes the
system's Hamiltonian to change with time. 
 We indicate the trajectory through the space of Hamiltonians as $H_{0..1} := ( H_0, H_{\Delta \tau}, \ldots, H_1 )$.

Note that the units of time are arbitrary and $\Delta \tau$ can 
be arbitrarily small (though non-zero). 
Accordingly our results hold exactly no matter how long
the process takes, and in particular even in the quasi-static limit.
The choice of countable state space and discretized time is used to simplify analysis, in line with much of the literature~\cite{crooks1999entropy,hatano1999jarzynski,chernyak2006path,still2012thermodynamics}. However, our approach should extend to continuous state space and continuous time.

Write the distribution of the system's state at $t$ as $p_t(x)$, or equivalently $p(x_t)$.
Due to thermal fluctuations and driving by the work reservoir, the system undergoes a stochastic dynamics,
represented by a
conditional distribution of trajectories given initial states, $p(x_{0..1}\vert x_{0})$ (we make no assumptions about whether this  dynamics is first-order Markovian or not).
The conditional distribution over trajectories in turn induces a conditional distribution of final states given initial states,
$p(x_1 \vert x_0) = \sum_{x'_{0..1}} \delta_{x'_1, x_1} p(x'_{0..1} \vert x_0)  $, which we sometimes refer to as
a \textbf{map} that takes initial states $x_0$ to final states $x_1$.

We refer to a given pair of $H_{0..1}$ and $p(x_{0..1}\vert x_{0})$
as a \textbf{(thermodynamic) process} operating on the system, indicated generically as $\proc$.
Note that any process  $\proc$ can be prepared with different initial distributions $p_0$, giving different trajectory probabilities $p(x_{0..1}):=p(x_{0..1}|x_0)p_0(x_0)$. 

Given a sequence of Hamiltonians $H_{0..1}$, the total work done on the system if it follows state trajectory $x_{0..1}$ is
\begin{align}
 W(x_{0..1}) = \sum_{t\in \{0,\Delta \tau,..,1\}} H_{t +\Delta \tau}(x_t) - H_t(x_t) \,.
\end{align}
For an initial distribution $p_0$, the expected work across all trajectories is
\[
\langle W \rangle_{p_0} = \sum_{x_{0..1}} p_0(x_0)p(x_{0..1}|x_0) W(x_{0..1}) \,.
\]

Suppose we seek to drive the system from some particular (possibly non-equilibrium)
 initial distribution $p_{0}$ to some final distribution $p_{1}$, while changing the Hamiltonian from $H_{0}$ to $H_{1}$.
 Define the \textbf{non-equilibrium free energy}~\cite{parrondo2015thermodynamics} of a system with Hamiltonian $H_t$ and distribution $p_t(x)$ as
\[
\mathcal{F}\left(H_t,p_t\right):= \langle H_t \rangle_{p_t} - kT \cdot S(p_t) \,,
\]
where $S(p) := -\sum_x p(x) \ln p(x)$ indicates Shannon entropy (in nats). (Note that $\mathcal{F}$ is equal to the equilibrium free energy when $p_t$ is the Boltzmann distribution for Hamiltonian $H_t$.) For any 
process $\proc$ that transforms $(p_0,H_0) \rightarrow (p_1,H_1)$, expected work is lower bounded by
\begin{equation}
\langle W \rangle_{p_0} \ge \mathcal{F}\left(H_{1},p_{1}\right)-\mathcal{F}\left(H_{0},p_{0}\right)  \,.
\label{eq:deltafreeenergy}
\end{equation}
This inequality reflects the modern understanding of the second law~\cite{takara2010generalization,deffner_information_2012,parrondo2015thermodynamics,esposito2011second}.%

The difference of non-equilibrium free energies is called the \textbf{reversible work}. %
Reversible work is the portion of expected work that could be recovered from the heat bath and system after the process finishes, by transforming  the system from $H_1,p_1$ back to $H_0,p_0$ in a thermodynamically reversible manner (in this way completing a thermodynamic cycle).
Reversible work can be either positive or negative, depending on $H_0, p_0, H_1$ and $p_1$. 

\textbf{Dissipated work}, or simply \textbf{dissipation}, is the portion of expected work that cannot be thermodynamically recovered~\cite{parrondo2015thermodynamics,takara2010generalization,kawai_dissipation:_2007}. It is written as
\begin{equation}
W_{d}(p_0) := \langle W \rangle_{p_0} - \left[\mathcal{F}(H_{1},p_{1})-\mathcal{F}(H_{0},p_{0}) \right] \,.
\label{eq:diss-definition}
\end{equation}
The dissipation associated with a process 
is always non-negative, and it is zero iff the process is thermodynamically reversible.
(Dissipation should not to be confused with the dissipated heat, which is the total energy
transferred to the heat bath, nor with expected total work minus the change in \emph{equilibrium} free energies, which is also sometimes called dissipated work~\cite{kawai_dissipation:_2007,gomez2008footprints,parrondo_entropy_2009}.)

Define $\mathcal{Q}(x_0)$ %
as the expected total heat transferred from the bath
to the system during the process if the system starts in $x_0$~\cite{crooks1998nonequilibrium}. By conservation of energy we can
write this as
\[
\mathcal{Q}(x_0) := \sum_{x_{0..1}'} p(x_{0..1}'|x_0) ( H_1(x_1') - H_0(x_0) - W(x_{0..1}') )  \,,
\]
so that the total expected heat transferred is $ \langle \mathcal{Q}(X_0) \rangle_{p_0} = \sum_{x_0} p(x_0) \mathcal{Q}(x_0)$.
This allows us to rewrite dissipation as
\begin{align}
W_d(p_0) = kT [ S(p_1) - S(p_0)] - \langle \mathcal{Q}(X_0) \rangle_{p_0} \,,
\label{eq:diss-heat}
\end{align}
where $p_1(x') = \sum_{x} p(x_1 \vert x_0) p_0(x_0)$ is the final state distribution when the process is initialized with $p_0$.
Thus, dissipation is proportional to the entropy change that does not correspond to heat exchanged with the heat bath,
which is called the (irreversible) \textbf{entropy production}~\cite{esposito2011second,deffner_nonequilibrium_2011,parrondo2015thermodynamics}.

In the remainder of this paper we choose units so that $kT = 1$.

\section{Dissipation due to incorrect priors}
\label{sec:mainresult}

Let $q_0$ be an initial distribution that achieves minimum dissipation for a given $\proc$,
\begin{align}
q_0 := \argmin_{p_0} W_d(p_0)  \,.
\label{eq:q_def}
\end{align}
We call $q_0$ the \textbf{prior} distribution for $\proc$ (for reasons made clear below).
We do not assume that the prior distribution is unique. 

While $q_0$ is an initial distribution that results in minimal dissipation, in general
$\proc$ may be prepared with some initial distribution $r_0$, which we call the \textbf{environment distribution}, that need
not equal $q_0$. By definition, 
\[
W_d(r_0) - W_d(q_0) \ge 0 \,.
\]
We call this extra dissipation when using $r_0$ rather than $q_0$ the \textbf{incorrect prior dissipation}.
Notice that if $\proc$ achieves zero dissipation for some initial distribution, then 
$W_d(q_0)=0$ %
and dissipation and incorrect prior dissipation are equivalent.

Several papers have shown that it is possible to design a process that implements any given
stochastic map $p(x_1 \vert x_0)$ with zero dissipation for any given initial distribution $p_0$~\cite{maroney2009generalizing,wolpert2016free,wolpert_landauer_2016a}. 
Incorrect prior dissipation first appeared in these analyses: it was shown that a 
\emph{particular type of process} that implements a given stochastic map and achieves zero dissipation for a particular $q_0$ will dissipate work when prepared with a different initial distribution $r_0 \ne q_0$.
Here we generalize these previous analyses; the main result of our paper is a simple expression for incorrect prior dissipation that 
applies to \emph{any} thermodynamic process. 

To derive our main result, note that by definition, the prior $q_0$ minimizes the differentiable function $W_d$ over the set of all valid probability distributions.  We assume that $q_0$ has full support, i.e., it is in the interior of the unit simplex. (This assumption will often hold; \ref{appendix:prior-support-conditions} presents one particular sufficient condition concerning $p(x_1 \vert x_0)$.) Then,
for any initial state distribution $r_0$, the directional derivative at $q_0$ must obey
\begin{align}
(r_0 - q_0) \cdot \nabla W_d(q_0) = 0 \,,
\label{eq:optimality-condition}
\end{align}
where $\cdot$ indicates the dot product.

Next we use \eqref{eq:diss-heat} to write the $|X|$ components of $\nabla W_d(p_0)$,
\begin{align}
\frac{\partial W_d}{\partial p(x_0)}(p_0) &= \Big[-\sum_{x_{1}}p(x_{1}\vert x_{0})\ln \Big(\sum_{x_0'} p_0(x_0') p(x_1|x_0') \Big) - 1 \Big] +\Big[ \ln p(x_{0} ) + 1 \Big] -\mathcal{Q}(x_0) \nonumber \\
& = -\sum_{x_{1}}p(x_{1}\vert x_{0})\ln p_1(x_1) + \ln p_0(x_{0} ) - \mathcal{Q}(x_0) \,.
\label{eq:gradient_w_d}
\end{align}
Combining \eqref{eq:gradient_w_d} and \eqref{eq:diss-heat} lets us express the inner products as %
\begin{align}
q_0 \cdot \nabla W_d(q_0) &= 
S(q_1) - S(q_0) - \langle Q \rangle_{q_0} = W_d(q_0)
\label{eq:like_a_linear} \\
\begin{split}
r_0 \cdot \nabla W_d(q_0) & = 
C(r_1\Vert q_1) - C(r_0 \Vert q_0) - \langle Q \rangle_{r_0} \\
& = D(r_1\Vert q_1) - D(r_0 \Vert q_0) + W_d(r_0) \,.
\end{split} 
\label{eq:cross_entropies}
\end{align}
where $C(p \Vert q) := -\sum_x p(x) \ln q(x)$ is the cross-entropy function and $D(p\Vert q) = \sum_x p(x) \ln \frac{p(x)}{q(x)} = C(p\Vert q) - S(p)$ 
is the Kullback-Leibler (KL) divergence~\cite{cover_elements_2012}. 

Combining Eqs.~\ref{eq:optimality-condition}, \ref{eq:like_a_linear}, and \ref{eq:cross_entropies} leads to our main result: incorrect prior dissipation for any distribution $r_0$ is
\begin{align}
W_d(r_0) - W_d(q_0) = D(r_{0}\Vert q_{0})-D(r_{1}\Vert q_{1}) \,.
\label{eq:main-result}
\end{align}
(See \ref{appendix:general-statement} for an extension of this result for the case where all distributions
are restricted to a convex subset of the unit simplex.)

Recall that the KL divergence $D(r\Vert q)$ is an information-theoretic 
measure of the distinguishability of distributions $r$ and $q$~\cite{cover_elements_2012}. Thus, our main result states that
incorrect prior dissipation measures the decrease in our ability to distinguish
whether the initial distribution was $q_0$ or $r_0$ as the system evolves
from $t=0$ to $t=1$.  Formally, this drop reflects the ``{contraction of KL divergence}'' under the action of the map 
$p(x_1|x_0)$~\cite{ahlswede_spreading_1976,cohen_relative_1993}.
It is non-negative due to the KL data processing inequality~\cite[Lemma 3.11]{csiszar_information_2011}.
(This is consistent with our main result, since
incorrect prior dissipation measures extra dissipation relative to the minimum possible.)

Interestingly, the contraction of KL divergence reflects the logical reversibility of the map $p(x_{1}\vert x_{0})$.
If $p(x_{1}\vert x_{0})$ specifies a logically-reversible map
from $x_0$ to $x_1$ (i.e., a permutation over $X$), then incorrect prior dissipation is 0 for all $r_0$.
At the other extreme, if $p(x_1 \vert x_0)$ is an input-independent map, where changing $x_0$ has no
effect on the resultant distribution over $x_1$, then $D(r_1\Vert q_1) = 0$ and incorrect prior dissipation reaches its maximum value of $D(r_0\Vert q_0)$. In addition, in this case the prior distribution that minimizes $W_d(.)$
is unique, since $W_d(r_0) = D(r_0\Vert q_0)=0$ iff $r_0 = q_0$.
More generally, in \ref{appendix:positive-dissipation-for-noninvertible} we prove that if and only if $p\left(x_{1}\vert x_{0}\right)$
is not a logically reversible map, then there must exist an $r_0$ such that $W_d(r_0) - W_d(q_0) >0$ .

For another perspective on \eqref{eq:main-result},
note that by the chain rule for KL divergence \cite[Eq. 2.67]{cover_elements_2012}, 
\begin{align}
D (r(X_{0},X_{1}) \Vert q(X_{0},X_{1}))
 & =D(r_0 \Vert q_0)+D(r(X_{1}\vert X_{0})\Vert q(X_{1}\vert X_{0}))\nonumber \\
 & =D(r_1 \Vert q_1)+D(r(X_{0}\vert X_{1})\Vert q(X_{0}\vert X_{1})) 
 \label{eq:rewrite} \,.
\end{align}
However, since $r(x_{1}\vert x_{0})=q(x_{1}\vert x_{0})=p(x_{1}\vert x_{0})$,
$D(r(X_{1}\vert X_{0}) \Vert q(X_{1}\vert X_{0}))=0$.
Thus, \eqref{eq:main-result} is equivalent to 
$$
W_{d}(r_0) - W_d(q_0) = D(r(X_{0}\vert X_{1})\Vert q(X_{0}\vert X_{1})) \,.
$$
(See also \cite{wolpert_landauer_2016a}.)
In this expression $r(x_{0}\vert x_{1})$ and $q(x_{0}\vert x_{1})$
are Bayesian posterior probabilities of the initial state conditioned
on the final state, for the assumed priors $r_0$ 
and $q_0$ respectively, and the shared likelihood function $p(x_{1}\vert x_{0})$.
(This Bayesian formulation of \eqref{eq:main-result} is why we refer to the initial distribution $q_0$ as a ``prior''.)

In \ref{appendix:general-statement}, we show that if $q_0$ is \emph{not} assumed to have full support, then the RHS of \eqref{eq:main-result} becomes a lower bound (rather than an equality) on the incorrect prior dissipation.

\section{Discussion of incorrect prior dissipation}
In this section, we present some important implications and generalizations of our main result, as well
as some caveats that are important to keep in mind.

Note that a thermodynamic process $\proc$ is specified by a large set of real numbers: the values of the Hamiltonian $H_{0..1}$ and the conditional distribution $p(x_{0..1}\vert x_0)$. (In fact, in the $\Delta \tau \rightarrow 0$ limit this set is infinite.)
However, 
by \eqref{eq:diss-heat}, the dissipation function $W_d(\cdot)$
can be specified using only $|X|^2$ real numbers: the $|X|$ values of 
$\mathcal{Q}(x_0)$ 
and the $|X|(|X| - 1)$ values of $p(x_1 \vert x_0)$.
Unfortunately, the values $\mathcal{Q}(x_0)$ may be impractical to compute for a given $\proc$, 
since they involve expectations over a very large set of trajectories.  Indeed, the distribution
over trajectories may not even be fully specified if some details of the process are unknown.

\eqref{eq:main-result} shows that $W_d(\cdot)$  can alternatively be parameterized by the $|X|^2$ numbers: the value of $W_d(q_0)$, the $|X|-1$ values of $q_0$, and the $|X|(|X| - 1)$ values of $p(x_1 \vert x_0)$.
This also means that, perhaps surprisingly, 
calculating the amount of dissipation above the minimum possible 
only requires knowledge of the stochastic map $p(x_1 \vert x_0)$ and a minimizer $q_0$, 
and does not depend on any specifics of the intermediate process. Given some initial distribution $r_0$, all physical details of how  $\proc$ manages to transform $q_0 \rightarrow q_1$ and $r_0 \rightarrow r_1$ are irrelevant for evaluating incorrect prior dissipation.

It is important to emphasize that our analysis above does not specify how to
find the minimizer $q_0$.
In some cases, it may be possible to find $q_0$ via numerical minimization of the convex function $W_d(p_0)$ over a $|X|^2$-dimensional space. (See \ref{appendix:convexity} for a proof that $W_d$ is convex.)
In others, such minimization may be achievable via analytical techniques, or it may be possible to 
analytically find an initial distribution that achieves zero dissipation (which must then be a minimizer).  
Some previous studies have used these kinds of techniques to find priors $q_0$
and our results can provide additional insight into those studies.
For example, one published model of Maxwell's demon used numerical methods to derive an inequality for dissipated work~\cite[Eq. 10]{mandal2012work}.
As we show in \ref{appendix:mandal-and-jarzynski}, our results can be used in a straightforward manner to derive this inequality analytically --- and in fact provide an exact expression for dissipated work.
It is also important to emphasize that our main result concerns only one 
contributor to the total dissipated work (namely the amount in addition to the minimum amount possible).
Moreover, dissipated work itself is just one contributor to expected total work.
Thus, for instance, the fact that  incorrect prior dissipation is related to the logical irreversibility of the map $p(x_1|x_0)$ has no direct implications for whether \emph{total} dissipation and/or {total} work is small for
a thermodynamic process with a logically reversible map~\cite{logical_rev_relevance}.
In addition, note that the prior distribution $q_0$, which minimizes dissipation, will not generally be the initial distribution that 
minimizes expected total work.  Indeed, since total expect work is linear in the initial distribution over states,
the distribution that minimizes expected total work is a delta function about
$x_0^\star = \argmin_{x_0} \sum_{x_{0..1}} p(x_{0..1}|x_0) W(x_{0..1})$. In general, that delta function distribution
will not  minimize dissipation.

There are some conditions, however, when incorrect prior dissipation \emph{can} be related to expected total work. Consider the case when the process $\proc$ is thermodynamically-reversible for some initial distribution, meaning that $W_d(q_0) =0$. Then, the expected work when $\proc$ is
prepared with initial distribution $r_0$ is 
\begin{align}
\langle W \rangle_{r_0} & =  D(r_{0}\Vert q_{0})-D(r_{1}\Vert q_{1}) + \mathcal{F}\left(r_{1},H_{1}\right)-\mathcal{F}\left(r_{0},H_{0}\right) \nonumber \\
& = \langle H_1 \rangle_{r_1} - \langle H_0 \rangle_{r_0} + C(r_0 \Vert q_0) - C(r_1 \Vert q_1) \,.
\label{eq:cross_entropy}
\end{align}
If, furthermore, both the initial and final Hamiltonians $H_0$ and $H_1$ are uniform over the space of allowed states, then expected work for initial distribution $r_0$ is
$C(r_0 \vert \vert q_0) - C(r_1 \vert \vert q_1)$. (See \cite{wolpert2016free} for
an example of a physical system where this is the case.) 

Finally, it is possible to generalize our main result in two important ways, as shown in \ref{appendix:general-statement}.
First, when the minimizer $q_0$ does \emph{not} have full support, incorrect prior dissipation is lower-bounded by (rather than equal to) the contraction of KL divergence. In addition, our main result can be generalized to the case when $q_0$ is not the minimizer of $W_d$ over all possible initial distributions, but only within some convex subspace of distributions.  Then, the result holds for any other initial distribution $r_0$ within the same subspace. The latter generalization is used in the next section to derive a coarse-grained version of \eqref{eq:main-result}.

\section{Thermodynamics of computation}

We now extend our main result, to apply to the thermodynamics of computation. Formally,  this means that
we analyze the implications of our main result for physical systems that perform information-processing 
operations over some coarse-grained degrees of freedom.

Recent advances in nonequilibrium
statistical physics \cite{sagawa2014thermodynamic, parrondo2015thermodynamics}
have extended and clarified the pioneering analysis of of Landauer,
Bennett and others \cite{landauer1961irreversibility,bennett1982thermodynamics,zurek1989thermodynamic,zure89b} regarding the fundamental thermodynamics cost of information processing.
In this section, we consider the implications of incorrect prior dissipation for thermodynamics of computation.

In keeping with previous analyses, we define  a \textbf{computer} 
as a physical system with microstates $x \in X$ undergoing a thermodynamic process $\proc$, together with 
a coarse-graining of $X$ into a set of \textbf{Computational Macrostates} (CMs) with labels $v \in V$
(the set of CMs are equivalent to what are called the ``information bearing degrees of freedom'' in~\cite{bennett2003notes}, and the ``information states'' in~\cite{deffner2013information}). $\proc$
induces a stochastic dynamics over $X$, and the (possibly non-deterministic) {computation} is identified
with the associated dynamics over CMs.
We use $\pi(v_{1} \vert v_{0})$ to indicate this dynamical process 
over CMs, i.e., to indicate a single iteration of the computation that maps \textbf{inputs} $v_0$
to \textbf{outputs} $v_1$. 
The canonical example of this kind of computation is a single iteration of a laptop, modifying the bit pattern in its memory (i.e., its CM)~\cite{bennett2003notes}.  
In practice, computers are usually designed to perform the same
operation over their CMs from one iteration to the next.  Formally, this means that their dynamics are first-order Markovian and time-homogeneous.

In previous work \cite{wolpert_landauer_2016a}, we showed that for
any given $\pi$ and input distribution, a computer can be designed
that implements $\pi$ with zero dissipation for that input distribution. Here, we instead consider how the amount of dissipation for a fixed, given computer depends on the choice of input distribution.
We recover a coarse-grained version of \eqref{eq:main-result}, expressing incorrect prior dissipation for a distribution over input CMs. Thus, the exact same equations that determine how dissipation varies with the initial distribution over \emph{micro}states also determine how dissipation
in a computer varies with the initial distribution 
over computational \emph{macro}states.

Formally, let $g : X \rightarrow V$ be
the coarse-graining function that maps the microstates of a computer to its CMs.
Let $s(x \vert v)$ be a fixed distribution over the microstates corresponding to 
the specified macrostate $v$, and so obeys $s(x \vert v)=0$ if $v \ne g(x)$. We use the random variables $V_0$ and $V_1$ to indicate the CM at the beginning and end of the process, respectively.
To avoid confusion between distributions over CMs and those over microstates, distributions over CM are superscripted 
with a $V$. Thus, we write $p_0^V(.)$ and $p_1^V(.)$ to indicate the distribution over CMs at $t=0$ and
at $t=1$, respectively, and similarly for $q_0^V, q_1^V, r_0^V$ and $r_1^V$. 

When combined with the conditional update distribution $\pi(v_1 \vert v_0)$, any initial
distribution $p^V_0$ over CMs induces a final
distribution over states of $V$ at $t=1$ in the obvious way. Such a $p^V_0$ also induces a 
$t=0$ mixture distribution over microstates, given by
averaging the distributions $s(x \vert v)$ over all possible $v$. It will be useful to write this mixture with the shorthand
\begin{flalign}
[\Phi(p^V)](x)=\sum_{v} s(x|v)p^V(v) = s(x|g(x))p^V(g(x)) \,.
\label{eq:phi-def}
\end{flalign}
Thus, $\Phi$(.) is a map that takes distributions over $V$ to distributions over $X$.
The image of $\Phi$, $\imgPhi$, is a convex subset of the set of distributions over $X$, containing all possible mixtures of $s(x|v)$ induced by distributions over $V$.

We make two assumptions in our analysis of computers, 
which capture some physical properties of what is commonly meant by ``computers'', both in the real world and in the literature on thermodynamics of computation.

First, we assume the initial distribution over microstates is determined by specifying the initial distribution over
CMs. 
This assumption reflects the fact that in current real-world
computers, the input is set by selecting some computational macrostate (e.g., setting the pattern of logical bits in memory).
It is \emph{not} selected by the user selecting a particular microstate of the system (which would
occur, for example, if the user set the positions and momenta of all atoms and electrons in the computer).
Formally, this assumption means that any allowed initial distribution $p_0$ must be an element of $\imgPhi$, and will satisfy $p_0(x_0) = [\Phi(p^V_0)](x_0)$ for some initial distribution over CMs $p^V_0$. We call such an initial distribution over CMs an \textbf{input distribution}.

Second, we assume that
the distribution of microstates, conditioned on the respective macrostate, is the same at the beginning and end of the thermodynamic process. 
This assumption guarantees that the dynamics of a computer's logical state is first-order Markovian and time-homogeneous; in other words, the computer can be run for multiple iterations, and it is guaranteed to obey the same logical rules in each iteration.  
We formalize this assumption by requiring that the dynamics obey
\begin{align*}
p(x_1 \vert v_0, v_1) = \frac{p(x_1,v_1|v_0) }{ p(v_1|v_0)}= \frac{ \sum_{x_0} \delta_{v_1,g(x_1)}p(x_1 \vert x_0) s(x_0 | v_0)}{ p(v_1|v_0)} = s(x_1|v_1)
\end{align*}
for all $x_1, v_0, v_1$. 
In words, this states that $v_0$ is conditionally independent of $x_1$ given $v_1$ (i.e., there is no information about $v_0$ ``hidden'' in the microstate $x_1$ beyond that provided by the fact that $x_1$ belongs to CM $v_1$). This assumption also means that as long as the initial microstate distribution $p_0$ is induced by some distribution over CMs, then the output microstate distribution $p_1$ is also induced by some distribution over CMs (i.e., that $p_1 \in \imgPhi$ so long as $p_0 \in \imgPhi$).  %
We refer to any process that obeys this condition as \textbf{computationally cyclic}.

When the computer is run with the microstate distribution $\Phi(p_0^V)$, the amount of dissipated work is $W_d(\Phi(p_0^V))$.  
Accordingly we refer to $W_d(\Phi(p_0^V))$ as the \emph{dissipation of the (macrostate) input distribution $p_0^V$}, 
and when clear from context, write it simply
as  $W_d(p_0^V)$.

In analogy with the case of dynamics over
$X$, we say that an input distribution $q_0^V$ is a \textbf{prior} for the computer if it achieves minimum dissipation among all input distributions.  As we did in our analysis of priors over microstates, we assume that 
the prior $q_0^V$ has full support.
(More formally, see \ref{appendix:prior-support-conditions} for a sufficient condition on $\pi$ under which this assumption will hold.)

Let $q_0 := \Phi(q_0^V) \in \imgPhi$ indicate the microstate distribution induced by $q_0^V$. By definition of $q_0^V$,
$q_0$ has minimum dissipation within the convex set $\imgPhi$.  Furthermore,  by our assumption that $q_0^V$ has full support, $\Phi(q_0^V)$ will be in the relative interior of $\imgPhi$.

Now consider any other input distribution $r_0^V$, as well as its associated microstate distribution $r_0 := \Phi(r_0) \in \imgPhi$.  Using the general statement of dissipation due to incorrect priors derived in \ref{appendix:general-statement},
\begin{align*}
W_{d}(\Phi(r_0^V)) - W_d(\Phi(q_0^V)) & =D(r_0\Vert q_0)-D(r_1\Vert q_1) \\
& =  D(r(V_{0},X_{0})\Vert q(V_{0},X_{0}))-D(r(V_{1},X_{1})\Vert q(V_{1},X_{1}))\\
& =  D(r_0^V\Vert q_0^V)-D(r_1^V\Vert q_1^V) \\
& \quad  +D(r(X_{0}\vert V_{0})\Vert q(X_{0}\Vert V_{0}))-D(r(X_{1}\vert V_{1})\Vert q(X_{1}\vert V_{1})) \,,
\end{align*}
where the second line follows because $v_{0}$ and $v_{1}$ are deterministic
functions of $x_{0}$ and $x_{1}$, and the third line follows from the chain
rule for KL divergence. Next, note that by definition
$r(x_{0}\vert v_{0})=s(x_{0}\vert v_{0})=q(x_{0}\vert v_{0})$. So $D(r(X_{0}\vert V_{0})\Vert q(X_{0}\vert V_{0})) = 0$.
In addition, by the cyclic condition,
$
r(x_1|v_1)=\sum_{v_0} r(v_0|v_1)p(x_1|v_0,v_1)=s(x_1|v_1)
$
and similarly for $q(x_1|v_1)$. So $D(r(X_{1}\vert V_{1})\Vert q(X_{1}\vert V_{1})) = 0$.

Combining leads to a coarse-grained version of our main result: for dynamics over CMs, incorrect prior dissipation for any input distribution $r_0^V$ is
\begin{align*}
W_{d}(r_0^V) - W_d(q_0^V)  = D(r_0^V\Vert q_0^V)-D(r_1^V\Vert q_1^V) \,.
\end{align*}

The obvious analog of \eqref{eq:cross_entropy} (and the associated discussion) holds for computers,
if we replace distributions over microstates by distributions over CMs. 
These results agree with the analysis for a specific model of a computer in~\cite{wolpert_landauer_2016a}.
However the analysis here holds for any computer, no matter how it operates.

As before, if $q_0^V$ does not have full support, we recover an inequality rather than an equality (\ref{appendix:general-statement}).

\section{Conclusion} 
For a fixed nonequilibrium process,
we have quantified the additional dissipation arising from using some arbitrary initial distribution, relative to the dissipation incurred when using the initial distribution that achieves minimal dissipation. This additional dissipation has a simple, information-theoretic form, being equal to the the contraction of KL divergence between the actual and optimal initial distributions over the course of the process. 

We also considered computers, i.e., processes that implement some stochastic map over a set of coarse-grained variables. We showed that our main result applies to distributions over coarse-grained states of a system,
so long as the fine-grained dynamics obey several conditions.
Landauer and co-workers pioneered analysis of the thermodynamic cost of computation; in
its modern formulation, Landauer's bound considers the minimal (dissipation-free) total work needed to perform a given computation~\cite{parrondo2015thermodynamics,esposito2011second}.  Our result extends
these analyses to include the dissipation cost of computation, and in particular its dependence on the initial
distribution of the computer's states.

Our results are derived with few assumptions.
They do not require that the dynamics obey local detailed balance, nor that they are Markovian.
In addition, they hold for both quasi-static and finite time processes, and regardless of how far the process is from equilibrium.

\emph{Acknowledgments} \textemdash
We would like to thank the Santa Fe Institute for helping to support this research. This paper was made possible through the support of Grant No. TWCF0079/AB47 from the Templeton World Charity Foundation, Grant No.
FQXi-RFP-1622 from the FQXi foundation, and Grant No. CHE-1648973 from the U.S. National Science Foundation. 
The opinions expressed in this paper are those of the authors and do not necessarily 
reflect the view of Templeton World Charity Foundation.

\appendix

\section{Dissipation due to incorrect priors over convex spaces}
\label{appendix:general-statement}

Let $\Delta$ be a convex subset of the set of all distributions over state space $X$.  For a given process $\proc$, define the prior distribution in $\Delta$ as
$$q_0 := \argmin_{p_0 \in \Delta} W_d(\proc, p_0)\,. $$

\begin{proposition}
For all initial distributions $r_0 \in \Delta$,
\label{prop1-general}
\begin{align}
W_d(r_0) - W_d(q_0) \ge D(r_{0}\Vert q_{0})-D(r_{1}\Vert q_{1})\,,
\label{prop1-appendix-ineq}
\end{align}
where $D(\cdot\Vert \cdot)$ is the KL divergence. If the prior distribution $q_0$ is in the relative interior of $\Delta$, the inequality is tight:
\begin{align}
W_d(r_0) - W_d(q_0) = D(r_{0}\Vert q_{0})-D(r_{1}\Vert q_{1}) \,.
\label{prop1-appendix-eq}
\end{align}

\end{proposition}
\begin{proof}
First use \eqref{eq:diss-heat} in the main text to write the $|X|$ components of $\nabla W_d(p_0)$  as
\begin{align}
\frac{\partial W_d}{\partial p(x_0)}(p_0) & =
\Big[-\sum_{x_{1}}p(x_{1}\vert x_{0})\ln \sum_{x_0'} p_0(x_0') p(x_1|x_0') - 1 \Big] +\Big[ \ln p(x_{0} ) + 1 \Big] -\mathcal{Q}(x_0)  \nonumber \\
& = -\sum_{x_{1}}p(x_{1}\vert x_{0})\ln p_1(x_1) + \ln p_0(x_{0} ) - \mathcal{Q}(x_0) \,.
\label{eq:gradient_w_d-appendix}
\end{align}
Note that by \eqref{eq:gradient_w_d} and \eqref{eq:diss-heat} in the main text, even though $W_d (p_0)$ is not a linear function of $p_0$, it is still true that for any $p_0$,
\begin{align}
W_d(p_0) = \sum_{x_0} p(x_0)  \frac{\partial W_{d}}{\partial p(x_0)}(p_0) = p_0 \cdot \nabla W_d(p_0) \,.
\label{eq:like_a_linear-appendix}
\end{align}

The prior $q_0$ minimizes $W_d$ in the convex space $\Delta$.  Then, the directional derivative at $q_0$ toward $r_0 \in \Delta$, written as $(r_0 - q_0) \cdot \nabla W_d(q_0)$, must be non-negative, since otherwise $W_d$ could be decreased by slightly perturbing $q_0$ toward $r_0$. Thus,
\begin{align}
(r_0 - q_0) \cdot \nabla W_d(q_0) \ge 0 \,.  \nonumber
\end{align}
By \eqref{eq:like_a_linear-appendix},
\begin{align}
r_0 \cdot \nabla W_d(q_0) \ge W_d(q_0) \,. \nonumber
\end{align}
Plugging in~\eqref{eq:gradient_w_d-appendix}, we rewrite,
\begin{align}
r_0 \cdot \nabla W_d(q_0) & = -\langle \mathcal{Q}\rangle_{r_0} - C(r_0\Vert q_0) + C(r_1\Vert q_1) \nonumber \\
&= W_d(r_0) - \left[ D(r_0\Vert q_0) - D(r_1\Vert q_1) \right] \,,
\label{eq:generalproof-appendix-kl-drop}
\end{align}
where $C(\cdot \Vert \cdot)$ is the cross-entropy.
Combining establishes the inequality Prop.~\ref{prop1-appendix-ineq}.

If $q_0$ is in the relative interior of $\Delta$, the directional derivatives at $q_0$ must be positive toward \emph{and} away-from $r_0$. Thus,
\[
(r_0 - q_0) \cdot \nabla W_d(q_0)  \ge  0 \text{ and } (q_0 - r_0) \cdot \nabla W_d(q_0) \ge 0  \,,
\]
leading to 
\begin{align}
(r_0 - q_0) \cdot \nabla W_d(q_0)  =  0  \,.
\label{eq:generalproof-appendix-both-sides}
\end{align}
Combining \eqref{eq:generalproof-appendix-both-sides}, \eqref{eq:like_a_linear-appendix}, and \eqref{eq:generalproof-appendix-kl-drop} establishes the equality Prop.~\ref{prop1-appendix-eq}.
\end{proof}

\section{Dissipated work is convex}
\label{appendix:convexity}

First, consider two initial distributions, specified by the conditional probability distribution $w(X_0=x_0\vert C=0)$ and $w(X_0=x\vert C=1)$, as well as the mixture $w(X_0=x)=\sum_c p(C=c) \, w(X=x\vert C=c)$. At the end of the process, these distributions are mapped to $w(X_1=x|C=0)$, $w(X_1=x|C=1)$, and $w(X_1=x)=\sum_c p_C(C=c)\,w(X_1=x|C=c)$.

To demonstrate that $W_d$ is convex, we will show that
\begin{align*}
p(C=0) W_d(w(X_{0}\vert C=0))+p(C=1) W_d(w(X_{0}\vert C=1)) \ge W_d(w(X_0))\,.
\end{align*}

First, we subtract the RHS from the LHS, while using the expression for dissipated work (\eqref{eq:diss-heat}).
The linear terms drop out, leaving the entropy terms:
\begin{align*}
& \,p(C=0) W_d(w(X_{0}\vert C=0))+p(C=1) W_d(w(X_{0}\vert C=1))-W_d(w(X_0)) \\
=& \, p(C=0) \left[S(w(X_{1}\vert C=0)) - S(w(X_{0}\vert C=1)) \right] \\
& +p(C=1) \left[ S(w(X_{1}\vert C=1)  - S(w(X_{0}\vert C= 1))\right] \\
& -  \left[ S(w(X_{1})) - S(w(X_{0})) \right] \\
=&\, MI(X_{0};C)-MI(X_{1};C) \\
\ge &\, 0 \,.
\end{align*}
The last line follows from the data processing inequality for mutual information~\cite{cover_elements_2012}.

\section{Analysis of `Maxwell's demon' model of Mandal and Jarzynski} 
\label{appendix:mandal-and-jarzynski}

We consider the work dissipated in the thermodynamic process corresponding
to one ``interaction interval'' of the information-processing ``demon''
described in \cite{mandal2012work}. Let $X=\left\{ A0,B0,C0,A1,B1,C1\right\} $
represent the state space of the model. Also let $V$ be a coarse-graining
of $X$ into a binary state (corresponding to the state of the bit
on the tape), where $V=0$ corresponds to $\left\{ A0,B0,C0\right\} $
and $V=1$ corresponds to $\left\{ A1,B1,C1\right\} $. As in our main text,
it will useful to distinguish distributions over $V$ from those over $X$
with a superscript, e.g., writing $r^V_0$ rather than $r_0$.

The model is parameterized by:
\begin{enumerate}
\item $\tau$: the amount of time the demon interacts with each incoming
bit, i.e., the length of a single interaction interval. In our framework,
this means that $t\in\left[0,1\right]$ maps to a duration of physical
time $\tau$.
\item $\delta$: set the `excess' of 0s in the incoming bit distribution (i.e., the distribution
of $V$ at the beginning of the interaction interval), via $\delta=r_{0}^{V}\left(V=0\right)-r_{0}^{V}\left(V=1\right)$.
\item $\epsilon$: set the `excess' of 0 in the equilibrium distribution of outgoing bits (i.e.,
the distribution of $V$ at the end of an interaction interval as
$\tau\rightarrow\infty$), via $\epsilon = p_\text{eq}^{V}\left(V=0\right)-p_\text{eq}^{V}\left(V=1\right)$
\cite[Eq. 6b]{mandal2012work}.
\end{enumerate}
The parameter $\epsilon$ is used to define a continuous-time $\left|X\right|\times\left|X\right|$
rate matrix $\mathscr{R}$ \cite[Eq. S1]{mandal2012work} specifying system
dynamics. This rate matrix is then used to define a transition
matrix $\Pi^{\tau}:=e^{\tau\mathscr{R}}$ for interactions of duration $\tau$.%

In \cite{mandal2012work}, it is noted that dissipation is 0 when
$\delta=\epsilon$. In their Supporting Information, the authors provide
a complex derivation showing that dissipation is non-negative for
other cases. This is shown analytically for the quasi-static limit of interval lengths
($\tau\rightarrow\infty$, i.e. when each interaction interval takes
an infinite amount of time), but only numerically for finite interval lengths
\cite[Eq. S20]{mandal2012work}. Here we show how to use the results in
our paper to prove \emph{strict }positivity simply, and analytically,
for all time scales.

Note that $\mathscr{R}$ is irreducible, and hence has a unique stationary
distribution, which we call $p_\text{eq}^X\left(x\right)$ ($p_\text{eq}^V$ is a marginalization of this stationary distribution onto the $V$ subspace). When the 6-state system
is prepared with initial distribution $p^X_\text{eq}$, no work gets done
\cite{mandal2012work} and the nonequilibrium free energy doesn't
change, hence $W_{d}\left(p_\text{eq}^X\right)$ is 0. Therefore, in the language of our main text, $p_{eq}$ is a prior distribution for this thermodynamic process, since no other initial distribution can achieve lower dissipation.

Using our main result, we write dissipation when the system is prepared with initial distribution $r_0^X$ and allowed to interact for duration $\tau$ as
\begin{align}
W_{d}\left(r_{0}^X\right)-W_{d}\left(p_\text{eq}^X\right) & =W_{d}\left(r_{0}^X\right) \nonumber \\
& =D\left(r_{0}^X \Vert p_\text{eq}^X\right)-D\left(\Pi^{\tau}r_{0}^X\Vert\Pi^{\tau} p_\text{eq}^X\right)\label{eq:diss-statement-mj}\\
& =D\left(r_{0}^X\Vert p_\text{eq}^X\right)-D\left(\Pi^{\tau}r_{0}^X\Vert p_\text{eq}^X\right)\nonumber \\
& >0\;\;\;\;\text{whenever }r_{0}^X\ne p_\text{eq}^X  \nonumber 
\end{align}
where the inequality arises from the fact that irreducible rate matrices
have strict convergence to equilibrium~\cite[Section 3.5]{greven2003entropy}.

Note that $\epsilon=\delta$ means that $r_{0}^{V}\left(v\right)=p_\text{eq}^{V}(v)$.
Thus $\epsilon=\delta$ is a necessary condition for $r_{0}^X(x)=p_\text{eq}^X(x)$
(though not sufficient, since we would also need $r_{0}(x|v)=p_\text{eq}(x|v)$).
We have thus shown that $\epsilon=\delta$ is a necessary condition for
dissipation to be 0, and that when $\epsilon\ne\delta$, dissipation
is guaranteed to be strictly positive.

Observe also that for any specific values $\tau,\epsilon,\delta$,
and $r_{0}\left(x\vert v\right)$, we can use Eq. \ref{eq:diss-statement-mj}
above to compute dissipation exactly.

\section{Sufficient conditions for prior to have full support} 
\label{appendix:prior-support-conditions}

In this section of the SM we assume that all components of $p(x_{1}\vert x_{0})$ are nonzero
and that $\mathcal{Q}(x_{0})$ is finite for all $x_{0}$, and show that this
means that both minimizers $q(x_0)$ and $q(v_0)$ have full support.

To begin, expand~\eqref{eq:gradient_w_d} in the main text to write
\begin{align}
\frac{\partial W_{d}}{\partial p(x_{0})}(p_{0})= %
-\mathcal{Q}(x_{0})-\sum_{x_{1}}p(x_{1}\vert x_{0})\ln\left[\frac{\sum_{x_{0}'}p(x_{1}\vert x_{0}')p(x_{0}')}{p(x_{0})}\right]
\label{eq:diss-deriv}
\end{align}
for any distribution $p_0$ over $X$.

Define $q_0 := \argmin_{p_0 \in \Delta} W_d(p_0)$, where $\Delta$ is the $|X|$-dimensional unit simplex.  To show that $q_0$ has full support, hypothesize
that there exists some $x_{0}^{\star}$ such that $q(x_{0}^{\star})=0$.
Now consider the one-sided derivative 
$\frac{\partial W_{d}}{\partial p(x_{0}^{\star})}(q_{0})$.
By the assumption that  $p\left(x_{1}\vert x_{0}\right)>0$ for all $x_0, x_1$,
the numerator inside the logarithm in \eqref{eq:diss-deriv} is nonzero, while
by hypothesis the denominator is 0. Thus, the argument of the logarithm
is positive infinite and (since $\mathcal{Q}(x_{0}^{\star})$ is finite, by assumption)
$\frac{\partial W_{d}}{\partial p(x_{0}^{\star})}(q_{0})$
is negative infinite. Moreover, for any $x'_{0}$ where $q(x'_{0})>0$,
$\frac{\partial W_{d}}{\partial p_{0}(x_{0}^{'})}(q_{0})$
is finite. This means that $W_{d}(q_{0})$ can be reduced by increasing
$q(x_{0}^{\star})$ and (to maintain normalization) reducing $q(x'_{0})$,
contrary to the definition of $q_{0}$ as a minimizer. Therefore our
hypothesis must be wrong.

Next, consider the prior input distribution $q_0^V := \argmin_{p_0^V \in \Delta^V} W_d(\Phi(p_0^V))$, where $\Delta^V$ is the $|V|$-dimensional unit simplex.
To show that $q_0^V$ has full support under the above assumptions, consider the partial derivative of dissipation wrt to each entry of the input probability distribution, $\frac{\partial  W_d(\Phi({p}_V))}{\partial p_V(v_{0})}$.
Let $p_0 := \Phi(p_0^V)$, and then use the chain rule, \eqref{eq:phi-def}, \eqref{eq:diss-deriv}, and then \eqref{eq:phi-def}  in the main text again to write
\begin{align}
\frac{\partial  W_d(\Phi({p}_V))}{\partial p_V(v_{0})}\left(p_0^V\right) & =\sum_{x_{0}}\frac{\partial W_{d}}{\partial p(x_{0})} (p_0) \; \frac{\partial [\Phi({p}_V)](x_{0})}{\partial p_V(v_{0})}\left(p_0^V\right) \nonumber \\
& =\sum_{x_{0}}\frac{\partial W_{d}}{\partial p(x_{0})} (p_0) \; s\left(x_{0}\vert v_{0}\right) \nonumber \\
& =\sum_{x_{0}}\left[-\mathcal{Q}(x_{0})-\sum_{x_{1}}p(x_{1}\vert x_{0})\ln\frac{\sum_{x_{0}'}p(x_{1}\vert x_{0}')p(x_{0}')}{p(x_{0})}\right]s\left(x_{0}\vert v_{0}\right) \nonumber\\
& =\sum_{x_{0}}\left[-\mathcal{Q}(x_{0})-\sum_{x_{1}}p(x_{1}\vert x_{0})\ln\frac{\sum_{x_{0}'}p(x_{1}\vert x_{0}')p(x_{0}')}{s\left(x_{0}\vert v_{0}\right)p_V(v_{0})}\right]s\left(x_{0}\vert v_{0}\right)
\label{eq:diss-deriv-v}
\end{align}
for any distribution $p_V^0$ over $V$.

Proceeding as before, hypothesize that there exists some $v_{0}^{\star}$
such that $q_V(v_{0}^{\star}) = 0$, and 
use~\eqref{eq:diss-deriv-v} to evaluate 
$\frac{\partial W_{d}(\Phi(p_V))}{\partial p(v_{0}^{\star})}(q_0^V)$.
By our hypothesis, the associated value of the denominator in the logarithm in~\eqref{eq:diss-deriv-v} is zero.
Since $p(x_1 \vert x_0)$ is always nonzero by assumption, this means the sum over $x_1$ is positive infinite.
Since by assumption $\mathcal{Q}(x_0)$ is bounded, this means
that $\frac{\partial W_{d}(\Phi(p_V))}{\partial p(v_{0}^{\star})}(q_0^V)$
is negative infinite. At the same time, $\frac{\partial W_{d}(\Phi(p_V))}{\partial p(v_{0}')}(q_0^V)$
is finite for any $v_{0}'$ where $q_V(v'_{0})>0$. Thus
$W_{d}(\Phi(q_{0}^V))$ can be reduced by increasing $q_V(v_{0}^{\star})$
and (to maintain normalization) reducing $q_V(v'_{0})$, contrary
to the definition of $q_0^V$ as a minimizer. Therefore our hypothesis
must be wrong.

\section{Proof of strictly positive dissipation for non-invertible maps\label{appendix-show-dissipation}}
\label{appendix:positive-dissipation-for-noninvertible}

Suppose the driven dynamics $p\left(x_{1}\vert x_{0}\right)$ is a
stochastic map from $X \rightarrow X$ that results
in minimal dissipation for some prior distribution $q_0$.
\begin{thm}
Suppose that $q_0$ has full support. Then, there exists $r_0$
with incorrect prior dissipation $W_d(r_0) - W_d(q_0) >0$ iff $p\left(x_{1}\vert x_{0}\right)$
is not an invertible map.
\end{thm}
\begin{proof}

If $q_0$ has full support, then $\supp r_0 \subseteq \supp q_0$ for all $r_0$. Then \eqref{eq:main-result} states that if initial distribution $r_0$ is used, extra dissipation is equal to
\begin{align}
W_d(r_0) - W_d(q_0) & =D(r_0\Vert q_0)-D(r_1\Vert q_1) \nonumber \\
&  =D(r(X_{0}\vert X_{1})\Vert q(X_{0}\vert X_{1})) 
\label{eq:condKL}
\end{align}
KL divergence is invariant under invertible transformations.  Therefore, if $p\left(x_{1}\vert x_{0}\right)$ is an invertible map, then $D(r_0\Vert q_0) = D(r_1\Vert q_1) \implies W_d(r_0) - W_d(q_0) = 0$  $\forall \; r_0$.

We now prove that if $p\left(x_{1}\vert x_{0}\right)$ is not an invertible
map, then there exists $r_0$ such that $W_{d}(r_0) - W_d(q_0) >0$.
For simplicity, write the dynamics $p\left(x_{1}\vert x_{0}\right)$
as the right stochastic matrix $M$. Because $M$ is a right stochastic
matrix, it has a right (column) eigenvector $\mathbf{1}^{T}=(1,\dots,1)^{T}$
with eigenvalue 1. 

Furthermore, it is known that if $M$ is not an invertible map, i.e. permutation matrix, then $|\det M| < 1$~\cite{goldberg1966upper}.  Since the determinant is the product of the eigenvalues and the magnitude of any eigenvalue of a stochastic matrix is upper bounded by 1, $M$ must have at least one eigenvalue $\lambda$ with $|\lambda| < 1$.  Let $\mathbf{s}$ represent the non-zero left eigenvector corresponding to $\lambda$. Note that due to biorthgonality of eigenvectors, $\mathbf{s}\mathbf{1}^{T}=0$.
We use $s\left(x\right)$ to refer to elements of $\mathbf{s}$ indexed
by $x\in X$. Without loss of generality, assume $\mathbf{s}$
is scaled such that $\max_{x}\left|s\left(x\right)\right|=\min_{x_{0}}q\left(x_{0}\right)$ (which is greater than 0, by assumption that $q_0$ has full support).

We now define $r_0$ as
\[
r\left(x_{0}\right):=q\left(x_{0}\right)+s\left(x_{0}\right) 
\]
Due to the scaling of $\mathbf{s}$ and because $\mathbf{s}\mathbf{1}^{T}=0$,
$r_0$ is a valid probability distribution.

We use the notation $s\left(x_{1}\right):=\sum_{x_{0}}s\left(x_{0}\right)p\left(x_{1}\vert x_{0}\right)$
and $r\left(x_{1}\right):=\sum_{x_{0}}r\left(x_{0}\right)p\left(x_{1}\vert x_{0}\right)=q\left(x_{1}\right)+s\left(x_{1}\right)$.
We also use the notation $\mathcal{C}:=\supp r_1$.
The fact that $q_0$ has full support also means that $\mathcal{C} \subseteq \supp q_1$.

The proof proceeds by contradiction. Assume that $W_d(r_0) - W_d(q_0) =0$. Using \eqref{eq:condKL} and due to properties of KL divergence, this means that for each $x_{0}\in X$ and $x_{1}\in\mathcal{C}$,
\begin{align*}
q\left(x_{0}\vert x_{1}\right) & =r\left(x_{0}\vert x_{1}\right)\\
\frac{q\left(x_{0}\right)p\left(x_{1}\vert x_{0}\right)}{q\left(x_{1}\right)} & =\frac{r\left(x_{0}\right)p\left(x_{1}\vert x_{0}\right)}{r\left(x_{1}\right)}\\
\frac{r\left(x_{1}\right)}{q\left(x_{1}\right)}p\left(x_{1}\vert x_{0}\right) & =\frac{r\left(x_{0}\right)}{q\left(x_{0}\right)}p\left(x_{1}\vert x_{0}\right)\\
\frac{q\left(x_{1}\right)+s\left(x_{1}\right)}{q\left(x_{1}\right)}p\left(x_{1}\vert x_{0}\right) & =\frac{q\left(x_{0}\right)+s\left(x_{0}\right)}{q\left(x_{0}\right)}p\left(x_{1}\vert x_{0}\right)\\
\frac{s\left(x_{1}\right)}{q\left(x_{1}\right)}p\left(x_{1}\vert x_{0}\right) & =\frac{s\left(x_{0}\right)}{q\left(x_{0}\right)}p\left(x_{1}\vert x_{0}\right)\\
s\left(x_{1}\right)q\left(x_{0}\vert x_{1}\right) & =s\left(x_{0}\right)p\left(x_{1}\vert x_{0}\right) 
\end{align*}
Taking absolute value of both sides gives
\begin{align*}
\left|s\left(x_{1}\right)\right|q\left(x_{0}\vert x_{1}\right)=\left|s\left(x_{0}\right)\right|p\left(x_{1}\vert x_{0}\right) 
\end{align*}
Summing over $x_{0}\in X$ and $x_{1}\in\mathcal{C}$,
\begin{align}
\sum_{x_{1}\in\mathcal{C}}\sum_{x_{0} \in X}\left|s\left(x_{1}\right)\right|q\left(x_{0}\vert x_{1}\right) & =\sum_{x_{1}\in\mathcal{C}}\sum_{x_{0} \in X}\left|s\left(x_{0}\right)\right|p\left(x_{1}\vert x_{0}\right)\nonumber \\
\sum_{x_{1}\in X}\left|s\left(x_{1}\right)\right|-\sum_{x_{1}\notin\mathcal{C}}\left|s\left(x_{1}\right)\right| & =\sum_{x_{1} \in X}\sum_{x_{0} \in X}\left|s\left(x_{0}\right)\right|p\left(x_{1}\vert x_{0}\right) \nonumber \\
& \;\;\;\;- \sum_{x_{1}\notin\mathcal{C}}\sum_{x_{0} \in X}\left|s\left(x_{0}\right)\right|p\left(x_{1}\vert x_{0}\right)  \label{eq:summing-both-sides}
\end{align}
Note that for all $x_{1}\notin\mathcal{C}$, $r\left(x_{1}\right)=0$,
meaning that $s\left(x_{1}\right)=-q\left(x_{1}\right)$. Thus, 
$$\sum_{x_{1}\notin\mathcal{C}}\left|s\left(x_{1}\right)\right|=\sum_{x_{1}\notin\mathcal{C}}q\left(x_{1}\right) $$

Furthermore, for all $x_{1}\notin\mathcal{C}$, $r\left(x_{1}\right)=\sum_{x_{0}}r\left(x_{0}\right)p\left(x_{1}\vert x_{0}\right)=0$.
Thus, for all $x_{0}\in X$ where $p\left(x_{1}\vert x_{0}\right)>0$
for some $x_{1}\notin\mathcal{C}$, $r\left(x_{0}\right)=0$, meaning
$s\left(x_{0}\right)=-q\left(x_{0}\right)$. This allows us to rewrite
the last term in \eqref{eq:summing-both-sides} as
\begin{flalign*}
&\sum_{x_{1}\notin\mathcal{C}}\sum_{x_{0} \in X}\left|s\left(x_{0}\right)\right|p\left(x_{1}\vert x_{0}\right) \\
&\;\;=\sum_{x_{1}\notin\mathcal{C}}\sum_{x_{0}:p\left(x_{1}\vert x_{0}\right)>0}\left|s\left(x_{0}\right)\right|p\left(x_{1}\vert x_{0}\right) 
\\
&\;\;=\sum_{x_{1}\notin\mathcal{C}}\sum_{x_{0}:p\left(x_{1}\vert x_{0}\right)>0}q\left(x_{0}\right)p\left(x_{1}\vert x_{0}\right) \\
&\;\;=\sum_{x_{1}\notin\mathcal{C}}q\left(x_{1}\right) 
\end{flalign*}

Cancelling terms that equal $\sum_{x_{1}\notin\mathcal{C}}q\left(x_{1}\right)$
from both sides of Eq. \eqref{eq:summing-both-sides}, we rewrite
\begin{align}
\sum_{x_{1}}\left|s\left(x_{1}\right)\right| & =\sum_{x_{1}}\sum_{x_{0}}\left|s\left(x_{0}\right)\right|p\left(x_{1}\vert x_{0}\right)=\sum_{x_{0}}\left|s\left(x_{0}\right)\right| 
\label{eq:l1-eq}
\end{align}
In matrix notation, Eq. \eqref{eq:l1-eq} states that
\begin{equation}
\left\Vert \mathbf{s}M\right\Vert _{1}=\left\Vert \mathbf{s}\right\Vert _{1} 
\label{eq:l1norm-eq}
\end{equation}
where $\left\Vert \cdot\right\Vert _{1}$ indicates the vector $\ell_{1}$
norm. However, by definition $\mathbf{s}M=\lambda\mathbf{s}$. Hence,
\[
\left\Vert \mathbf{s}M\right\Vert _{1}=\left\Vert \lambda\mathbf{s}\right\Vert _{1}=\left|\lambda\right|\left\Vert \mathbf{s}\right\Vert _{1}<\left\Vert \mathbf{s}\right\Vert _{1} 
\]
meaning that Eq. \eqref{eq:l1norm-eq} cannot be true and the original
assumption $W_d(r_0) - W_d(q_0) =0$ is incorrect. We have shown that
for non-invertible maps, there always exists an $r_0$ for which
$W_d(r_0) - W_d(q_0) >0$.
\end{proof}

\section*{References}
\bibliographystyle{iopart-num}
\bibliography{unified}

\providecommand{\newblock}{}
\begin{thebibliography}{10}
\expandafter\ifx\csname url\endcsname\relax
  \def\url#1{{\tt #1}}\fi
\expandafter\ifx\csname urlprefix\endcsname\relax\def\urlprefix{URL }\fi
\providecommand{\eprint}[2][]{\url{#2}}

\bibitem{touchette2004information}
Touchette H and Lloyd S 2004 {\em Physica A: Statistical Mechanics and its
  Applications\/} {\bf 331} 140--172

\bibitem{sagawa2009minimal}
Sagawa T and Ueda M 2009 {\em Physical review letters\/} {\bf 102} 250602

\bibitem{dillenschneider2010comment}
Dillenschneider R and Lutz E 2010 {\em Physical review letters\/} {\bf 104}
  198903

\bibitem{sagawa2012fluctuation}
Sagawa T and Ueda M 2012 {\em Physical review letters\/} {\bf 109} 180602

\bibitem{crooks1999entropy}
Crooks G~E 1999 {\em Physical Review E\/} {\bf 60} 2721

\bibitem{crooks1998nonequilibrium}
Crooks G~E 1998 {\em Journal of Statistical Physics\/} {\bf 90} 1481--1487

\bibitem{chejne2013simple}
Chejne~Janna F, Moukalled F and G{\'o}mez C~A 2013 {\em International Journal
  of Thermodynamics\/} {\bf 16} 97--101

\bibitem{jarzynski1997nonequilibrium}
Jarzynski C 1997 {\em Physical Review Letters\/} {\bf 78} 2690

\bibitem{esposito2011second}
Esposito M and Van~den Broeck C 2011 {\em EPL (Europhysics Letters)\/} {\bf 95}
  40004

\bibitem{esposito2010three}
Esposito M and Van~den Broeck C 2010 {\em Physical Review E\/} {\bf 82} 011143

\bibitem{parrondo2015thermodynamics}
Parrondo J~M, Horowitz J~M and Sagawa T 2015 {\em Nature Physics\/} {\bf 11}
  131--139

\bibitem{pollard2016second}
Pollard B~S 2016 {\em Open Systems \& Information Dynamics\/} {\bf 23} 1650006

\bibitem{wiesner2012information}
Wiesner K, Gu M, Rieper E and Vedral V 2012 {\em Proceedings of the Royal
  Society A: Mathematical, Physical and Engineering Science\/} {\bf 468}
  4058--4066

\bibitem{still2012thermodynamics}
Still S, Sivak D~A, Bell A~J and Crooks G~E 2012 {\em Physical review
  letters\/} {\bf 109} 120604

\bibitem{prokopenko2013thermodynamic}
Prokopenko M, Lizier J~T and Price D~C 2013 {\em Entropy\/} {\bf 15} 524--543

\bibitem{prokopenko2014transfer}
Prokopenko M and Lizier J~T 2014 {\em Nature Scientific reports\/} {\bf 4}

\bibitem{dunkel2014thermodynamics}
Dunkel J 2014 {\em Nature Physics\/} {\bf 10} 409--410

\bibitem{roldan2014universal}
Rold{\'a}n {\'E}, Martinez I~A, Parrondo J~M and Petrov D 2014 {\em Nature
  Physics\/}

\bibitem{berut2012experimental}
B{\'e}rut A, Arakelyan A, Petrosyan A, Ciliberto S, Dillenschneider R and Lutz
  E 2012 {\em Nature\/} {\bf 483} 187--189

\bibitem{hasegawa2010generalization}
Hasegawa H~H, Ishikawa J, Takara K and Driebe D 2010 {\em Physics Letters A\/}
  {\bf 374} 1001--1004

\bibitem{takara2010generalization}
Takara K, Hasegawa H~H and Driebe D 2010 {\em Physics Letters A\/} {\bf 375}
  88--92

\bibitem{deffner_information_2012}
Deffner S and Lutz E 2012 {\em arXiv preprint arXiv:1201.3888\/}

\bibitem{seifert_entropy_2005}
Seifert U 2005 {\em Physical review letters\/} {\bf 95} 040602

\bibitem{jarzynski_rare_2006}
Jarzynski C 2006 {\em Physical Review E\/} {\bf 73} 046105

\bibitem{seifert2012stochastic}
Seifert U 2012 {\em Reports on Progress in Physics\/} {\bf 75} 126001

\bibitem{cover_elements_2012}
Cover T~M and Thomas J~A 2012 {\em Elements of information theory\/} (John
  Wiley \& Sons)

\bibitem{mack03}
Mackay D 2003 {\em Information Theory, Inference, and Learning Algorithms\/}
  (Cambridge University Press)

\bibitem{schmiedl2007optimal}
Schmiedl T and Seifert U 2007 {\em Physical review letters\/} {\bf 98} 108301

\bibitem{sivak2012thermodynamic}
Sivak D~A and Crooks G~E 2012 {\em Physical review letters\/} {\bf 108} 190602

\bibitem{aurell2012refined}
Aurell E, Gaw{\c{e}}dzki K, Mej{\'\i}a-Monasterio C, Mohayaee R and
  Muratore-Ginanneschi P 2012 {\em Journal of statistical physics\/} {\bf 147}
  487--505

\bibitem{funo_work_2016}
Funo K, Shitara T and Ueda M 2016 {\em Physical Review E\/} {\bf 94} 062112

\bibitem{mandal2012work}
Mandal D and Jarzynski C 2012 {\em Proceedings of the National Academy of
  Sciences\/} {\bf 109} 11641--11645

\bibitem{kawai_dissipation:_2007}
Kawai R, Parrondo J~M~R and Van~den Broeck C 2007 {\em Physical review
  letters\/} {\bf 98} 080602

\bibitem{hatano1999jarzynski}
Hatano T 1999 {\em Physical Review E\/} {\bf 60} R5017

\bibitem{chernyak2006path}
Chernyak V~Y, Chertkov M and Jarzynski C 2006 {\em Journal of Statistical
  Mechanics: Theory and Experiment\/} {\bf 2006} P08001

\bibitem{gomez2008footprints}
Gomez-Marin A, Parrondo J and Van~den Broeck C 2008 {\em EPL (Europhysics
  Letters)\/} {\bf 82} 50002

\bibitem{parrondo_entropy_2009}
Parrondo J~M, Van~den Broeck C and Kawai R 2009 {\em New Journal of Physics\/}
  {\bf 11} 073008

\bibitem{deffner_nonequilibrium_2011}
Deffner S and Lutz E 2011 {\em Physical Review Letters\/} {\bf 107} ISSN
  0031-9007, 1079-7114

\bibitem{maroney2009generalizing}
Maroney O 2009 {\em Physical Review E\/} {\bf 79} 031105

\bibitem{wolpert2016free}
Wolpert D~H 2016 {\em Entropy\/} {\bf 18} 138

\bibitem{wolpert_landauer_2016a}
Wolpert D~H 2016 Extending landauer's bound from bit erasure to arbitrary
  computation arXiv:1508.05319 [cond-mat.stat-mech]

\bibitem{ahlswede_spreading_1976}
Ahlswede R and G{\'a}cs P 1976 {\em The Annals of Probability\/}  925--939

\bibitem{cohen_relative_1993}
Cohen J~E, Iwasa Y, Rautu G, Beth~Ruskai M, Seneta E and Zbaganu G 1993 {\em
  Linear Algebra and its Applications\/} {\bf 179} 211--235 ISSN 0024-3795

\bibitem{csiszar_information_2011}
Csiszar I and K{\"o}rner J 2011 {\em Information theory: coding theorems for
  discrete memoryless systems\/} (Cambridge University Press)

\bibitem{logical_rev_relevance}
Indeed, though logical reversibility and thermal reversibility were associated
  in early work on the thermodynamics of computation, it is now understood that
  they are independent. For instance, one can design a process to erase a bit
  in a thermodynamically reversible manner even though bit-erasure is
  logically-irreversible~\cite{sagawa2014thermodynamic}, assuming that the
  distribution over the states of the bit is exactly known to the designer.

\bibitem{sagawa2014thermodynamic}
Sagawa T 2014 {\em Journal of Statistical Mechanics: Theory and Experiment\/}
  {\bf 2014} P03025

\bibitem{landauer1961irreversibility}
Landauer R 1961 {\em IBM journal of research and development\/} {\bf 5}
  183--191

\bibitem{bennett1982thermodynamics}
Bennett C~H 1982 {\em International Journal of Theoretical Physics\/} {\bf 21}
  905--940

\bibitem{zurek1989thermodynamic}
Zurek W~H 1989 {\em Nature\/} {\bf 341} 119--124

\bibitem{zure89b}
Zurek W~H 1989 {\em Phys. Rev. A\/} {\bf 40}(8) 4731--4751

\bibitem{bennett2003notes}
Bennett C~H 2003 {\em Studies In History and Philosophy of Science Part B:
  Studies In History and Philosophy of Modern Physics\/} {\bf 34} 501--510

\bibitem{deffner2013information}
Deffner S and Jarzynski C 2013 {\em Physical Review X\/} {\bf 3} 041003

\bibitem{greven2003entropy}
Greven A, Keller G and Warnecke G 2003 {\em Entropy\/} (Princeton University
  Press)

\bibitem{goldberg1966upper}
Goldberg K 1966 {\em J Res Nat Bur Stand Sect B\/}  157

\end{thebibliography}

\end{document}